\documentclass[12pt,twoside]{article}

\usepackage{amssymb}
\usepackage{amsmath}
\usepackage{textcomp}
\usepackage{cite}
\usepackage{ocg}
\usepackage{amsthm}
\usepackage{comment}

\newcommand{\bra}[1]{\langle #1|}
\newcommand{\ket}[1]{|#1\rangle}

\newcommand{\cent}[0]{\mbox{\textcent}}
\newcommand{\dollar}[0]{\$}

\newcommand{\rightstate}[1]{\overrightarrow{#1}}
\newcommand{\stopstate}[1]{\mspace{0mu}\downarrow\mspace{-5mu} #1}

\newtheorem{theorem}{Theorem}
\newtheorem{lemma}{Lemma}
\newtheorem{corollary}{Corollary}

\newtheorem{conjecture}{Conjecture}
\newtheorem{question}{Question}

\begin{document}

\title{SUPERIORITY OF ONE-WAY AND REALTIME QUANTUM MACHINES AND NEW DIRECTIONS\footnote{This work was
partially supported by T\"{U}B\.ITAK with grant 108E142 and FP7 FET-Open project QCS.}} 
\author{Abuzer Yakary{\i}lmaz}
\institute{
University of Latvia, Faculty of Computing, Raina bulv. 19, Riga, LV-1586, Latvia\\
Email: \texttt{abuzer@lu.lv}
}

\maketitle

\begin{abstract}
In automata theory, the quantum computation has been widely examined for finite state machines,
known as quantum finite automata (QFAs), and
less attention has been given to the QFAs augmented with counters or stacks.
Moreover, to our knowledge, there is no result related to QFAs having more than one input head.
In this paper, we focus on such generalizations of QFAs whose input head(s) operate(s) in 
one-way or realtime mode and present many superiority of them to their classical counterparts.
Furthermore, we propose some open problems and conjectures 
in order to investigate the power of quantumness better.
We also give some new results on classical computation.
\\ \\
\textbf{keywords:} quantum computation, randomization, quantum automata, pushdown automaton, blind counter automaton, multihead finite automaton, nondeterminism, bounded error
\end{abstract}

\thispagestyle{empty}

% AAAAAAAAAAAAAAAAAAAAAAAAAAAAAAAAAAAAAAAAAAAAAAAAAAAAAAAAAAAAAAAAAAAAAAAAAAAAAAAA % 
% AAAAAAAAAAAAAAAAAAAAAAAAAAAAAAAAAAAAAAAAAAAAAAAAAAAAAAAAAAAAAAAAAAAAAAAAAAAAAAAA % 

% SSSSSSSSSSSSSSSSSSSSSSSSSSSSSSSSSSSSSSSSSSSSSSSSSSSSSSSSSSSSSSSSSSSSSSSSSSSSSSSS %
% SSSSSSSSSSSSSSSSSSSSSSSSSSSSSSSSSSSSSSSSSSSSSSSSSSSSSSSSSSSSSSSSSSSSSSSSSSSSSSSS %
% SSSSSSSSSSSSSSSSSSSSSSSSSSSSSSSSSSSSSSSSSSSSSSSSSSSSSSSSSSSSSSSSSSSSSSSSSSSSSSSS %
\section{Introduction}
% SSSSSSSSSSSSSSSSSSSSSSSSSSSSSSSSSSSSSSSSSSSSSSSSSSSSSSSSSSSSSSSSSSSSSSSSSSSSSSSS %
% SSSSSSSSSSSSSSSSSSSSSSSSSSSSSSSSSSSSSSSSSSSSSSSSSSSSSSSSSSSSSSSSSSSSSSSSSSSSSSSS %
% SSSSSSSSSSSSSSSSSSSSSSSSSSSSSSSSSSSSSSSSSSSSSSSSSSSSSSSSSSSSSSSSSSSSSSSSSSSSSSSS %

Quantum computation is a generalization of classical computation \cite{Wa09,YS11A}.
Therefore, it is interesting to investigate the cases in which quantum computation is superior to 
classical computation.
In automata theory,
many superiority results have been obtained mostly for quantum finite automata\footnote{They are
also known as quantum Turing machines with constant space.} (QFAs) \cite{KW97,AF98,BC01B,NIH01,AW02,BP02,MP02,BMP05,YS09C,SY10A,YS10A,YS10B,SY11A,YS11A,YS11B} 
and a few for QFAs with counters \cite{Kr99,BFK01,YKTI02,YKI05,YFSA11A,SY11C} and for QFAs with a stack \cite{Go00,NIH01,MNYW05,NHK06}.

In this paper, we present many new results about how the quantumness and in some cases randomness 
adds power to the one-way and realtime computational models, 
i.e. multihead finite and pushdown automata, counter automata, etc.
Then, we present some open problems and conjectures for the further investigations.
We also give some new results about classical computation.

Due to their restricted  definitions, 
early QFA models and their variants 
were shown to be less powerful than their classical counterparts for many cases \cite{KW97,AF98,MC00,ANTV02,YKI05}.
In fact, these models do not reflect the full power of quantum computation \cite{Wa09B}.
Therefore, we use ``modern" definitions for the quantum models (e.g. \cite{Hi08,AY11A}).

After a concise background given in Section 2, we present our results in Section 3,
in which we classify the results under four subsections:
(3.1) nondeterminism, 
(3.2) blind counter automata,
(3.3) multihead finite automata, and
(3.4) multihead pushdown automata.

% SSSSSSSSSSSSSSSSSSSSSSSSSSSSSSSSSSSSSSSSSSSSSSSSSSSSSSSSSSSSSSSSSSSSSSSSSSSSSSSS %
% SSSSSSSSSSSSSSSSSSSSSSSSSSSSSSSSSSSSSSSSSSSSSSSSSSSSSSSSSSSSSSSSSSSSSSSSSSSSSSSS %
% SSSSSSSSSSSSSSSSSSSSSSSSSSSSSSSSSSSSSSSSSSSSSSSSSSSSSSSSSSSSSSSSSSSSSSSSSSSSSSSS %
\section{Background}
% SSSSSSSSSSSSSSSSSSSSSSSSSSSSSSSSSSSSSSSSSSSSSSSSSSSSSSSSSSSSSSSSSSSSSSSSSSSSSSSS %
% SSSSSSSSSSSSSSSSSSSSSSSSSSSSSSSSSSSSSSSSSSSSSSSSSSSSSSSSSSSSSSSSSSSSSSSSSSSSSSSS %
% SSSSSSSSSSSSSSSSSSSSSSSSSSSSSSSSSSSSSSSSSSSSSSSSSSSSSSSSSSSSSSSSSSSSSSSSSSSSSSSS %

We specifically give the definitions of three models in order to trace the proofs presented in the paper:
generalized finite automaton, one-way quantum finite automaton,
and realtime quantum automaton with one-blind counter.
The quantum models are defined based on a generic template that is given in Subsection (3.2).
We refer the reader
to \cite{Ro66,FMR68,Gr78,Si06,HS10} for the definitions of classical machines;
to \cite{Hi08,YS11A} for the definitions of the QFAs generalizing their classical counterparts; and,
to \cite{NC00} for a standard reference of quantum computation.

Throughout the paper, we use the following notations:
$ \Sigma $ not containing $ \cent $ and $ \dollar  $ (the left and right end-markers)
denotes the input alphabet; 
$ \tilde{\Sigma} = \Sigma \cup \{\cent,\dollar\} $;
$ Q $ is the set of (internal) states;
$ q_{0} \in Q $ is the initial state;
$ Q_{a} \subseteq Q $ is the set of accepting states;
$ \delta $ is the transition function;
$ f_{\mathcal{M}}(w) $ is the accepting probability (or value) of machine $ \mathcal{M} $ on $ w $;
$ w_{i} $ is the $ i^{th} $ symbol of $ w $;
$ |w| $ is the length of $ w $;
$ |w|_{\sigma} $ is the number of the occurrence of $ \sigma $ in $ w $,
where $ w \in \Sigma^{*} $.
The list of abbreviations used for models is given below:
\begin{itemize}
	\item the prefixes ``1" and ``rt" stand for \textit{one-way}\footnote{The input head(s) is (are) 
		not allowed to move to left.} 
		and \textit{realtime}\footnote{The input head(s) is (are) allowed to move only to the right.}
		input head(s), respectively;
	\item the letters ``D", ``N", ``P", and ``Q" used after ``1" or ``rt" stand for
		\textit{deterministic}, \textit{nondeterministic}, \textit{probabilistic}, and \textit{quantum}, respectively;
	\item the abbreviations ``FA", ``$ k $FA", ``PDA", ``PD$ k $A", and ``$ k $BCA" stand for
		\textit{finite automaton}, \textit{finite automaton with $ k $ input heads},
		\textit{pushdown automaton},
		\textit{pushdown automaton with $ k $ input heads}, and
		\textit{automaton with $ k $ blind counter(s)}, respectively,
		where $ k > 0 $.
\end{itemize}

For all models (except GFAs), the input $ w \in \Sigma^{*} $
is placed on a read-only two-way infinite tape as $ \tilde{w} = \cent w \dollar $ 
from the squares indexed by 1 to $ |\tilde{w}| $.
At the beginning, the head(s) is (are) initially placed on the square indexed by 1 and
the value(s) of the counter(s) is (are) set to zero.

\subsection{Generalized finite automaton} 

A generalized finite automaton (GFA) \cite{Tu69} is formally a 5-tuple 
\begin{equation*}
	\mathcal{G}=(Q,\Sigma,\{A_{\sigma \in \Sigma} \},v_{0},f),
\end{equation*}
where (i) $ A_{\sigma \in \Sigma} $'s are $ |Q| \times |Q| $-dimensional
real valued transition matrices, and,
(ii) $ v_{0} $ and $ f $ are real valued  \textit{initial} (column) and \textit{final} (row) vectors, respectively.
For an input string, $ w \in \Sigma^{*} $, the acceptance value of $ w $ associated by 
$ \mathcal{G} $ is defined as
\begin{equation*}
	f_{\mathcal{G}}(w)=f A_{w_{|w|}} \cdots A_{w_{1}} v_{0}.
\end{equation*}

\subsection{Generic templete for quantum machines}

Now, we briefly describe a general framework for quantum machines allowing to implement general quantum operators
(see \cite{Ya11A,YS11A} for details).
Each quantum machine has a special component, 
a finite register, not considered as a part of the configurations, with alphabet $ \Omega $ having
a distinguished symbol $ \omega_{1} $ (the initial symbol).
In each step of the transition, 
(i) the register is reset to $ \ket{\omega_{1}} $;
(ii) as a part of the transition, a symbol is written on the register; and,
(iii) the finite register is discarded.
For one-way models, we have a set of outcomes $ \Delta = \{ a,r,n \} $ 
($ \Omega $ is partitioned into there pairwise disjoint subsets, i.e. $ \Omega_{\tau \in \Delta} $)
and, before discarding process,
a projective measurement is applied on the register.
That is, $ P = \{ P_{\tau \in \Delta} \mid P_{\tau} = \sum_{\omega \in \Omega_{\tau}} \ket{\omega} \bra{\omega} \} $,
and the following actions are performed with respect to the outcomes:
(a) or (r) the computation is halted and the input is accepted or rejected, respectively,
(n) the computation continues.
For realtime models, the decision on the input is given after reading the whole input
by a projective measurement, applied on the space spanned by the internal states,
i.e. $ P = \{ P_{a},I-P_{a} \mid P_{a} = \sum_{q \in Q_{a}} \ket{q}\bra{q} \} $.
For the models with blind counters, 
an additional measurement is done on the counters to check whether their values are zero or not.

\begin{figure}[h!]
	\centering
	\begin{minipage}{0.5\textwidth}
		\small
		\begin{equation*}
		\begin{array}{ccccc}
			\multicolumn{1}{c|}{} & c_{1} & 
			\multicolumn{2}{c}{\ldots} & \multicolumn{1}{c|}{ c_{|\mathcal{C}|} } \\
			\hline
			\multicolumn{1}{c|}{c_{1}} & & & & \multicolumn{1}{c|}{} \\
			\multicolumn{1}{c|}{\vdots} & \multicolumn{4}{c|}{ E_{\omega_{1}} } \\
			\multicolumn{1}{c|}{c_{|\mathcal{C}|}} & & & & \multicolumn{1}{c|}{} \\
			\hline	
			\multicolumn{1}{c|}{c_{1}} & & & & \multicolumn{1}{c|}{} \\
			\multicolumn{1}{c|}{\vdots} & \multicolumn{4}{c|}{ E_{\omega_{2}} } \\
			\multicolumn{1}{c|}{c_{|\mathcal{C}|}} & & & & \multicolumn{1}{c|}{} \\
			\hline\multicolumn{1}{c|}{} & & & & \multicolumn{1}{c|}{} \\
			\multicolumn{1}{c|}{} & & & & \multicolumn{1}{c|}{} \\
			\multicolumn{1}{c|}{\vdots} & \multicolumn{4}{c|}{\vdots } \\
			\multicolumn{1}{c|}{} & & & & \multicolumn{1}{c|}{} \\
			\hline	
			\multicolumn{1}{c|}{c_{1}} & & & & \multicolumn{1}{c|}{} \\
			\multicolumn{1}{c|}{\vdots} & \multicolumn{4}{c|}{ E_{\omega_{|\Omega|}} } \\
			\multicolumn{1}{c|}{c_{|\mathcal{C}|}} & & & & \multicolumn{1}{c|}{} \\
			\hline	
		\end{array}
		\end{equation*}
	\end{minipage}
	\caption{Matrix $ \mathsf{E} $}
	\label{figure:matrix-E}
\end{figure}

A quantum machine operates on the space spanned by its configurations.
The computation begins with the initial configuration and continues until terminated.
The transitions between the configurations are determined by the transition function.
Let $ \mathcal{C}_{\mathcal{M}}^{w} $, shortly $ \mathcal{C} $, 
be the configuration set of $ \mathcal{M} $ for a given 
input $ w \in \Sigma^{*} $.
All transitions of $ \mathcal{M} $ on $ w $ can be summarized 
as in Figure \ref{figure:matrix-E}, in which $ E_{\omega \in \Omega} $ represents
all transitions between the configurations when $ \omega $ is written on the register.
To be a well-formed machine, for all $ w \in \Sigma^{*} $, 
the columns of the matrix $ \mathsf{E} $ (Figure \ref{figure:matrix-E})
form an orthonormal set\footnote{
	In fact, matrix $ \mathsf{E} $ (Figure \ref{figure:matrix-E}) is a part of a bigger unitary matrix, 
	say $ \mathsf{U} $,
	that is responsible for the evolution of configuration space joint with the finite register. 
	Since the register is reset to $ \ket{ \omega_{1} } $  in each step, only a part of $ \mathsf{U} $,
	which is exactly $ \mathsf{E} $, operates on the configuration space.
	Therefore, the columns of $ \mathsf{E} $ must be orthonormal.
}, i.e. equivalently, 
\begin{equation*}
	\sum_{\omega \in \Omega} E_{\omega}^{\dagger} E_{\omega} = I,
\end{equation*}
where $ E_{\omega}[j,i] $ denotes the amplitude of transition from the $ i^{th} $ configuration to 
$ j^{th} $ configuration by writing $ \omega $ on the register.

\subsection{One-way quantum finite automata} 

A one-way quantum finite automata (1QFA) \cite{YS11A} is a 7-tuple 
\begin{equation*}
	\mathcal{M} = (Q,\Sigma,\Omega,\delta,q_{0},\Omega_{a},\Omega_{r}),
\end{equation*}
where $ \Omega_{a} $ ($ \Omega_{r} $) is the set of accepting (rejecting) symbols.
The transition of $ \mathcal{M} $ is specified as:
\begin{equation}	
	\delta(q,\sigma) = \alpha (p,d,\omega)  ~~ (\alpha \in \mathbb{C}),
	\label{eq:1QFA-delta}
\end{equation}
where $ \mathcal{M} $, which is in state $ q \in Q $ and reads $ \sigma \in \Sigma $ on the input tape,
changes the internal state to $ p \in Q $, update the position of the input head with respect to 
$ d \in \{\downarrow,\rightarrow \} $, and  writes $ \omega \in \Omega $ on the finite register
with amplitude $ \alpha $.
For simplicity, we assume that range component $ d $ can be determined by component $ p $,
denoted $ \stopstate{p} $ and $ \rightstate{p} $, and 
so term $ d $ can be dropped from Equation \ref{eq:1QFA-delta}.
For a given string $ w \in \Sigma^{*} $,
the configurations of $ \mathcal{M} $ are the pairs
of $ (q,x) \in Q \times \{1,\ldots,|\tilde{w}|\} $ and $ (q_{1},1) $ is the initial one,
where $ x $ stands for the head position.

\subsection{Realtime quantum automaton with 1 blind counter} 

A realtime quantum automata with one blind counter (rtQ1BCA) is a 6-tuple
 \begin{equation*}
	\mathcal{M} = (Q,\Sigma,\Omega,\delta,q_{0},Q_{a}).
\end{equation*}
We assume that $ \mathcal{M} $ can have the capability\footnote{
	It is a well-known fact that (e.g. see \cite{YFSA11A})
	for any classical or quantum counter automata having the capability of updating its counter(s)
	from the set $ \{-m,\ldots,m\} $, there exists an equivalent counter automaton updating 
	its counter(s) from the set $ \{-1,0,1\} $ for any $ m>1 $.
} of updating its counter(s) from the set $ \{-m,\ldots,m\} $ for any fixed $ m>1 $.
The transition of $ \mathcal{M} $ is specified as:
\begin{equation}
	\label{eq:rtQ1BCA-delta}
	\delta(q,\sigma) = \alpha (p,c,\omega) ~~ (\alpha \in \mathbb{C}),
\end{equation}
where $ \mathcal{M} $, which is in state $ q \in Q $ and reads $ \sigma \in \Sigma $ on the input tape,
changes the internal state to $ p \in Q $, update the counter value by $ c \in \{-m,\ldots,m\} $, 
and  writes $ \omega \in \Omega $ on the finite register
with amplitude $ \alpha $.
For a given string $ w \in \Sigma^{*} $, 
the configurations of $ \mathcal{M} $ are the pairs of 
$ (q,v) \in Q \times \mathbb{Z} $ and $ (q_{1},0) $ is the initial one,
where $ v $ stands for the value of the counter.

\subsection{Language recognition}

The language recognition criteria used in the paper can be defined as follows:
\begin{itemize}
	\item A language $ L \subseteq \Sigma^{*} $ is said to be recognized by $ \mathcal{M} $ with error bound 
		$ \epsilon \in (0,\frac{1}{2}) $ if 
		(i) $ f_{\mathcal{M}}(w) \geq 1-\epsilon $ for $ w \in L $ and
		(ii) $ f_{\mathcal{M}}(w) \leq \epsilon $ for $ w \notin L $.
	\item A language $ L \subseteq \Sigma^{*} $ is said to be recognized by $ \mathcal{M} $ with 
		negative one-sided error bound $ \epsilon \in (0,1) $ if 
		(i) $ f_{\mathcal{M}}(w) = 1 $ for $ w \in L $ and
		(ii) $ f_{\mathcal{M}}(w) \leq \epsilon $ for $ w \notin L $.
	\item A language $ L \subseteq \Sigma^{*} $ is said to be recognized by $ \mathcal{M} $ in nondeterministic mode
		if  (i) $ f_{\mathcal{M}}(w) > 0 $ for $ w \in L $ and
		(ii) $ f_{\mathcal{M}}(w) = 0 $ for $ w \notin L $.
\end{itemize}
Note that, any negative one-sided error bound $ \frac{1}{2} $ can be easily converted 
to the general error bound $ \frac{1}{3} $ and 
any negative one-sided error bound in internal $ (\frac{1}{2},1) $ can be easily converted 
to a general error bound in interval $ (\frac{1}{3},\frac{1}{2}) $.
Moreover, as a special case, 
the class of languages recognized by rtQFAs in nondeterministic mode is denoted by NQAL \cite{YS10A}.

% SSSSSSSSSSSSSSSSSSSSSSSSSSSSSSSSSSSSSSSSSSSSSSSSSSSSSSSSSSSSSSSSSSSSSSSSSSSSSSSS %
% SSSSSSSSSSSSSSSSSSSSSSSSSSSSSSSSSSSSSSSSSSSSSSSSSSSSSSSSSSSSSSSSSSSSSSSSSSSSSSSS %
% SSSSSSSSSSSSSSSSSSSSSSSSSSSSSSSSSSSSSSSSSSSSSSSSSSSSSSSSSSSSSSSSSSSSSSSSSSSSSSSS %
\section{Main results}
% SSSSSSSSSSSSSSSSSSSSSSSSSSSSSSSSSSSSSSSSSSSSSSSSSSSSSSSSSSSSSSSSSSSSSSSSSSSSSSSS %
% SSSSSSSSSSSSSSSSSSSSSSSSSSSSSSSSSSSSSSSSSSSSSSSSSSSSSSSSSSSSSSSSSSSSSSSSSSSSSSSS %
% SSSSSSSSSSSSSSSSSSSSSSSSSSSSSSSSSSSSSSSSSSSSSSSSSSSSSSSSSSSSSSSSSSSSSSSSSSSSSSSS %

In our algorithms, we use a special kind of quantum transformation, 
\textbf{\textbf{$ N $-way QFT}} (quantum Fourier transform) \cite{KW97,YS09B,YFSA11A}.
 Let $ N > 1 $ be a integer. The $ N $-way QFT is the transformation
\begin{equation*}
	\delta(d_{j}) = \frac{1}{\sqrt{N}} 
		\sum\limits_{l = 1}^{N} e^{\frac{2 \pi i }{N}jl} (r_{l}), ~~~~ 1 \le j \le N ,
\end{equation*}
from the \textit{domain} elements  $ d_{1}, \ldots, d_{N} $ to the
\textit{range} (\textit{target}) elements $ r_{1}, \ldots, r_{N} $,
where $ r_{N} $ is the distinguished target elements.
The QFT  can be used to check whether separate computational paths
of a quantum program that are in superposition have converged to the
same configuration at a particular step. Assume that the program has
previously split to $N$ paths, say $ \mathsf{s_{j}} $ ($ 1 \leq j \leq N $), each of which have the same amplitude. 
We assume that $ s_{j} $ (when having $ d_{j} $) enters to $ s_{j,1},\ldots,s_{j,N} $ by the QFT.
If $ s_{j,l} \neq s_{j',l} $ for each $ j \neq j' $, then none of the target elements is interfered 
with each other and so the distinguished target exists with probability $ \frac{1}{N} $, where $ 1 \leq l \leq N $.
Otherwise, if each $ s_{j} $ makes the QFT in different computational steps,
then we obtain the same result. But, if all of them make the QFT simultaneously,
then all targets are interfered with each other and only the distinguished target survives
with probability 1.

% ssssssssssssssssssssssssssssssssssssssssssssssssssssssssssssssssssssssssssssssss %
% ssssssssssssssssssssssssssssssssssssssssssssssssssssssssssssssssssssssssssssssss %
\subsection{Nondeterminism} \label{sec:nondeterminism}
% ssssssssssssssssssssssssssssssssssssssssssssssssssssssssssssssssssssssssssssssss %
% ssssssssssssssssssssssssssssssssssssssssssssssssssssssssssssssssssssssssssssssss %

It was shown in \cite{YS10A} that $ L \subseteq \Sigma^{*} \in $ NQAL
if and only if $ L $ is defined by a GFA, say $ \mathcal{G} $, as follows: 
(i) $ f_{\mathcal{G}}(w) > 0  $ for $ w \in L $ and 
(ii) $ f_{\mathcal{G}}(w) = 0  $ for $ w \notin L $.
We show our results in this section based on this equivalence.

We already know that the class of languages recognized by 1NFAs (reps., 1NPDAs) 
is a proper subset of the class of languages recognized by rtQFAs (resp., 1QPDA) 
in nondeterministic mode \cite{YS10A,NHK06}.
We give a stronger version of these results by using noncontextfree language
$ L_{ijk} = \{ a^{i}b^{j}c^{k} \mid i \neq j , i \neq k, j \neq k, 0 \leq i,j,k \} $.

\begin{theorem}
	\label{thm:L-ijk}
	$ L_{ijk} $ is in NQAL.
\end{theorem}
\begin{proof}
	(See Appendix \ref{app:thm:L-ijk} for a complete proof)
	We can design a GFA to calculate the value of ($ |w|_{a}-|w|_{b} $) on a specified internal state.
	By tensoring this machine with itself, we can obtain the value of $ (|w|_{a}-|w|_{b})^{2} $.
	Similarly, we can also calculate the value of 
	\begin{equation*}
		(|w|_{a}-|w|_{b})^{2} (|w|_{a}-|w|_{c})^{2} (|w|_{b}-|w|_{c})^{2}.
	\end{equation*}
	Additionally, this value is multiplied by 0 if 
	the input is not of the form $ a^{+}b^{+}c^{+} $. 
	Therefore, the last result is a positive integer if $ w \in L_{ijk} $ and it is zero if $ w \notin L_{ijk} $.
\end{proof}

\begin{corollary}
	In nondeterministic mode, the class of the languages recognized by classical machines
	is a proper subset of the class of the languages recognized by quantum machines
	for any model between finite automaton and one-head pushdown automaton.
\end{corollary}

Now we give a separation result between deterministic and nondeterministic automata with blind counters.
Note that, every one-way versions of these models can be easily converted to ones operating in realtime.

\begin{theorem}
	\label{thm:rtDkBCA}
	If  $ L $ is recognized by a rtD$ k $BCA, then $ \overline{L} \in $NQAL, where $ k>0 $.
\end{theorem}
\begin{proof}
	(See Appendix \ref{app:thm:rtDkBCA} for a complete proof)
	Let $ \mathcal{D} $ be the rtD$ k $BCA recognizing $ L $.
	We can design a GFA, say $ \mathcal{G} $, to exactly mimic the state transitions of $ \mathcal{D} $.
	($ \mathcal{G} $ can also cover the transitions of $ \mathcal{D} $ on $ \cent $ and $ \dollar $,
	by its initial and final vector component.)
	During the simulation, $ \mathcal{G} $ additionally change the values of the internal states
	with respect to the following strategy: (let $ p_{i} $ be the $ i^{th} $ prime ($ 1 \le i \le k $))
	(i) at the beginning, the value of the internal state of $ \mathcal{G} $ corresponding to the initial 
	state of $ \mathcal{D} $ is 1, and,
	(ii) if the value of the $ i^{th} $ counter is updated by $ 1 $ (resp., $ -1 $), then 
	the value of the state is multiplied by $ p_{i} $ (resp., $ \frac{1}{p_{i}} $).

	Suppose that, the computation of $ \mathcal{D} $ ends in state $ q $ on input $ w $.
	Let $ q' $ be the internal state of $ \mathcal{G} $ corresponding to $ q $
	and $ c_{q'} $ be the value of $ q' $.
	It can be easily be verified that $ c_{q'} = 1 $
	if and only if all counters of $ \mathcal{D} $ are set to zeros at the end of the computation.
	The value of $ (c_{q'}-1)^{2} $ can also be calculated by a GFA
	$ \mathcal{G'} $ if $ q $ is an accepting state.
	So, for $ w \in L $ ($ w \notin L $), we have $ f_{\mathcal{G'}}(w) = 0 $
	($ f_{\mathcal{G'}}(w) > 0 $).
\end{proof}

In \cite{FYS10A}, it is shown that $ L_{say} $ cannot be recognized by a rtQFA with 
unbounded error (and so $ L_{say} \notin $ NQAL and $ \overline{L_{say}} \notin $ NQAL \cite{YS11A}),
where
\begin{equation*}
	L_{say} = \{ w \mid \exists u_{1},u_{2}, v_{1}, v_{2} \in \{a,b\}^{*},
		w=u_{1}bu_{2}=v_{1}bv_{2}, |u_{1}| = |v_{2}| \}.
\end{equation*}
However, it can be easily be shown that $ L_{say} $ can be recognized by a rtN$ 1 $BCA:
two $ b $'s (those can also be the same) can be selected nondeterministically and by using a blind counter, 
the lengths of the substrings before the first $ b $ and after the second $ b $ can be compared.

\begin{corollary}
	For any $ k \in \mathbb{Z}^{+} $, the class of languages recognized by 1D$ k $BCAs
	is a proper subset of the class of languages recognized by 1N$ k $BCAs.
\end{corollary}

% ssssssssssssssssssssssssssssssssssssssssssssssssssssssssssssssssssssssssssssssss %
% ssssssssssssssssssssssssssssssssssssssssssssssssssssssssssssssssssssssssssssssss %
\subsection{Blind counter automata}
% ssssssssssssssssssssssssssssssssssssssssssssssssssssssssssssssssssssssssssssssss %
% ssssssssssssssssssssssssssssssssssssssssssssssssssssssssssssssssssssssssssssssss %

\begin{lemma}
	\label{lem:upal-epsilon}
	For any $ \epsilon \in (0,\frac{1}{2}) $,
	$ L_{upal} = \{a^{n}b^{n} \mid n \geq 0 \} $ can be recognized by a rtQ1BCA
	with negative one-sided error bound $ \epsilon $.
\end{lemma}
\begin{proof}	
	Let $ N \geq 2 $ and $ \mathcal{M}_{upal,N} = (Q,\Sigma,\Omega,\delta,q_{0},Q_{a}) $ 
	be a rtQ1BCA, where
	$ Q = \{ q_{0},a_{0},r_{0} \} \cup \{ q_{j} \cup q'_{j} \cup p_{j} \cup r_{j} \mid 1 \leq j \leq N \} $,
	$ \Omega=\{ \omega_{1},\omega_{2},\omega_{r}  \} $,
	$ Q_{a} = \{a_{0},p_{N}\} $.
	The details of $ \delta $ is given in Figure \ref{fig:M-upal-N}. 
	(The missing part of $ \delta $ can be easily be completed.)

	We show that $ \mathcal{M}_{upal,N} $ recognizes $ L_{upal} $
	with negative one-sided error bound $ \frac{1}{N} $.
	Therefore, by setting $ N = \left\lceil \frac{1}{\epsilon} \right\rceil $,
	we obtain the desired machine.
	
	\begin{figure}[h!]	
	\centering
	\fbox{
	\scriptsize
	\begin{minipage}{1\textwidth}
		\begin{equation*}
			\begin{array}{|c|c|c|c|}
				\multicolumn{4}{c}{\mathsf{mainpath}}
				\\
				\hline
				\cent & a & b & \dollar
				\\ \hline 				
				\delta(q_{0},\cent)  = (q_{0},0,\omega_{1})
				&
				\begin{array}{lcl}
					\delta(q_{0},a) & = & \frac{1}{\sqrt{N}} \sum\limits_{j=1}^{N} (q_{j},j,\omega_{1}) \\
					\delta(r_{0},a) & = & (r_{0},0,\omega_{r})
				\end{array}
				& 
				\begin{array}{lcl}
					\delta(q_{0},b) & = & (r_{0},0,\omega_{1}) \\
					\delta(r_{0},b) & = & (r_{0},0,\omega_{r})
				\end{array}
				&
				\begin{array}{lcl}
					\delta(q_{0},\dollar) & = & (a_{0},0,\omega_{1})  \\
					\delta(r_{0},\dollar) & = & (r_{0},0,\omega_{r})
				\end{array}
				\\ \hline
			\end{array}
		\end{equation*}
		\begin{equation*}
			\begin{array}{|c|c|c|}
				\multicolumn{3}{c}{\mathsf{path_{j}} ~ (1 \leq j \leq N)} 
				\\ \hline
				a~\mbox{(\tiny before reading a $ b $)} & b & \dollar
				\\ \hline
					\begin{array}{lcl}
						\delta(q_{j},a) & = & (q_{j},j,\omega_{2}) \\						
					\end{array}
					&
					\begin{array}{lcl}
						\delta(q_{j},b) & = & (q'_{j},-j,\omega_{1}) \\
						\delta(q'_{j},b) & = & (q'_{j},-j,\omega_{2})
					\end{array}
					&
					\begin{array}{lcl}
						\delta(q_{j},\dollar) & = & (q_{j},0,\omega_{1}) \\
						\delta(q'_{j}\,\dollar) & = & 
							\frac{1}{\sqrt{N}} \sum\limits_{l = 1}^{N} 
							e^{\frac{2 \pi i }{N}jl} (p_{l},0,\omega_{1})		
					\end{array}
				\\ \hline
				\multicolumn{3}{|c|}{a~\mbox{(\tiny the first $ a $ after reading a $ b $)}}
				\\ \hline
				\multicolumn{3}{|c|}{ \delta(q'_{j},a)  =  (r_{j},j,\omega_{1}) }
				\\ \hline
			\end{array}
		\end{equation*}
		\begin{equation*}
			\begin{array}{|c|c|c|}
				\multicolumn{3}{c}{\mathsf{rejecting\mbox{-}path_{j}} ~ (1 \leq j \leq N)}
				\\
				\hline
				a & b & \dollar
				\\ \hline 				
				\delta(r_{j},a)  = (r_{j},0,\omega_{r})
				&
				\delta(r_{j},b)  = (r_{j},0,\omega_{r})
				& 
				\delta(r_{j},\dollar)  = (r_{j},0,\omega_{r})
				\\ \hline
			\end{array}
		\end{equation*}
	\end{minipage}} 
	\caption{The details of the transition function of $ \mathcal{M}_{upal,N} $}
	\label{fig:M-upal-N}
	\end{figure}		
		
	We begin with two trivial cases: 
	(i) if the input is empty string, then it is accepted with probability 1; 
	(ii) if the input begins with a $ b $, then it is rejected with probability 1.
	So, we assume the input to begin with an $ a $ in the remaining part.
	After reading the first $ a $, the computation is split into $ N $ different paths,
	$ \mathsf{path}_{j} $ ($ 1 \le j \leq N $), with amplitude $ \frac{1}{\sqrt{N}} $
	and the counter value is increased by $ j $ in $ \mathsf{path}_{j} $.
	Each path keeps the same increment strategy as long as reading $ a $'s.
	After reading a $ b $, each path switches to a decrement strategy such that 
	the counter value is decreased  by $ j $ in $ \mathsf{path}_{j} $
	as long as reading $ b $'s.
	
	If an $ a $ is read after a $ b $, $ \mathsf{path_{j}} $ passes to $ \mathsf{rejecting\mbox{-}path_{j}} $,
	in which the input is rejected with probability 1 at the end.
	Otherwise, the input is of the form $ a^{m}b^{n} $, where $ m > 0 $ and $ n \geq 0 $, and
	before reading $ \dollar $, the machine is in the following superposition (of the configurations):
	\begin{equation*}
		\sum_{j=1}^{N} \frac{1}{\sqrt{N}} \ket{ (q'_{j},j(m-n)) }.
	\end{equation*}
	Note that, if $ m=n $, then we have
	\begin{equation*}
		\sum_{j=1}^{N} \frac{1}{\sqrt{N}} \ket{ (q'_{j},0) }.
	\end{equation*}
	Thus, after reading $ \dollar $, each path enters an $ N $-way QFT.
	That is, (i) if $ m=n $, all configurations are interfered with each other
	and only $ \ket{(p_{N},0)} $ remains with probability 1 and so the input is accepted exactly;
	(ii) if $ m \neq n $, none of the configuration is interfered and so 
	the input is accepted with probability $ \frac{1}{N} $ --
	before the measurement,
	the configurations with $ p_{N} $ exist in the superposition as
	\begin{equation*}
		\sum_{j=1}^{N} \frac{1}{N} \ket{ (p_{N},j(m-n)) }.
	\end{equation*}
	Note that, in case of $ m \neq n $, the configurations with an internal state different than 
	$ q_{N} $ are observed with probability $ 1-\frac{1}{N} $ at the end.
\end{proof}

\begin{theorem}
	\label{thm:upal-star-epsilon}
	For any $ \epsilon \in (0,\frac{1}{2}) $,
	$ L_{upal}^{*} $ can be recognized by a rtQ1BCA with negative one-sided error bound $ \epsilon $.
\end{theorem}
\begin{proof}
	We use the idea presented in the proof of Lemma \ref{lem:upal-epsilon} after 
	making a small modification.
	Let $ N \geq 2 $ and $ \mathcal{M}_{upal^{*},N} = (Q,\Sigma,\Omega,\delta,q_{0},Q_{a}) $ 
	be a rtQ1BCA, where
	$ Q = \{ q_{0},a_{0},r_{0} \} \cup \{ q_{j} \cup q'_{j} \cup p_{j} \mid 1 \leq j \leq N \} 
	\cup \{ r_{j} \mid 1 \leq j \leq N-1 \} $,
	$ \Omega=\{ \omega_{1},\omega_{2},\omega_{3},\omega_{r}  \} $,
	$ Q_{a} = \{a_{0},p_{N}\} $.
	The details of $ \delta $ is given in Figure \ref{fig:M-upal-N-star}. 
	(The missing part of $ \delta $ can be easily be completed.)
	
	\begin{figure}[h!]	
	\centering
	\fbox{
	\scriptsize
	\begin{minipage}{1\textwidth}
		\begin{equation*}
			\begin{array}{|c|c|c|c|}
				\multicolumn{4}{c}{\mathsf{mainpath}}
				\\
				\hline
				\cent & a & b & \dollar
				\\ \hline 				
				\delta(q_{0},\cent)  = (q_{0},0,\omega_{1})
				&
				\begin{array}{lcl}
					\delta(q_{0},a) & = & \frac{1}{\sqrt{N}} \sum\limits_{j=1}^{N} (q_{j},j,\omega_{1}) \\
					\delta(r_{0},a) & = & (r_{0},0,\omega_{r})
				\end{array}
				& 
				\begin{array}{lcl}
					\delta(q_{0},b) & = & (r_{0},0,\omega_{1}) \\
					\delta(r_{0},b) & = & (r_{0},0,\omega_{r})
				\end{array}
				&
				\begin{array}{lcl}
					\delta(q_{0},\dollar) & = & (a_{0},0,\omega_{1})  \\
					\delta(r_{0},\dollar) & = & (r_{0},0,\omega_{r})
				\end{array}
				\\ \hline
			\end{array}
		\end{equation*}
		\begin{equation*}
			\begin{array}{|c|c|c|}
				\multicolumn{3}{c}{\mathsf{path_{j}} ~ (1 \leq j \leq N)} 
				\\ \hline
				a~\mbox{(\tiny before reading a $ b $)} & b & \dollar
				\\ \hline
					\begin{array}{lcl}
						\delta(q_{j},a) & = & (q_{j},j,\omega_{2}) \\
					\end{array}
					&
					\begin{array}{lcl}
						\delta(q_{j},b) & = & (q'_{j},-j,\omega_{1}) \\
						\delta(q'_{j},b) & = & (q'_{j},-j,\omega_{2}) 
					\end{array}
					&
					\begin{array}{lcl}
						\delta(q_{j},\dollar) & = & (q_{j},0,\omega_{1}) \\
						\delta(q'_{j}\,\dollar) & = & 
							\frac{1}{\sqrt{N}} \sum\limits_{l = 1}^{N} 
							e^{\frac{2 \pi i }{N}jl} (p_{l},0,\omega_{1}) 		
					\end{array}
				\\ \hline
				\multicolumn{3}{|c|}{a~\mbox{(\tiny the first $ a $ after reading a $ b $)}}
				\\ \hline
				\multicolumn{3}{|c|}{
					\delta(q'_{j},a)  =  \frac{1}{\sqrt{N}} \sum\limits_{l = 1}^{N-1} 
							e^{\frac{2 \pi i }{N}jl} (r_{l},0,\omega_{3})		
						+ \frac{1}{\sqrt{N}} e^{2 \pi i j}
						\left( \frac{1}{\sqrt{N}} \sum\limits_{k = 1}^{N} (q_{k},k,\omega_{3}) \right)
							}
				\\ \hline
			\end{array}
		\end{equation*}
		\begin{equation*}
			\begin{array}{|c|c|c|}
				\multicolumn{3}{c}{\mathsf{rejecting\mbox{-}path_{j}} ~ (1 \leq j \leq N-1)}
				\\
				\hline
				a & b & \dollar
				\\ \hline 				
				\delta(r_{j},a)  = (r_{j},0,\omega_{r})
				&
				\delta(r_{j},b)  = (r_{j},0,\omega_{r})
				& 
				\delta(r_{j},\dollar)  = (r_{j},0,\omega_{r})
				\\ \hline
			\end{array}
		\end{equation*}
	\end{minipage}} 
	\caption{The details of the transition function of $ \mathcal{M}^{*}_{upal,N} $}
	\label{fig:M-upal-N-star}
	\end{figure}
	
	Suppose that the input is of the form $ (a^{+}b^{+}) (a^{+}b^{+})^{+} $ or 
	$ (a^{+}b^{+}) (a^{+}b^{+})^{*}a^{+} $.
	(If not, $ \mathcal{M}_{upal^{*},N} $ behaves exactly the same as $ \mathcal{M}_{upal,N} $.)
	After reading the first block of $ (a^{+}b^{+}) $, 
	$ \mathcal{M}_{upal^{*},N} $ enters a QFT on the first $ a $ 
	to compare the number of $ a $'s and the number of $ b $'s in the block.
	The targets of the QFT are $ \mathsf{rejecting\mbox{-}path_{j}} $'s ($ 1 \leq j \leq N-1 $) 
	and the distinguished one in which the computation re-splits into $ \mathsf{path}_{k} $'s ($ 1 \leq k \leq N $)
	with equal amplitudes.
	Thus, if the block contains the equal number of $ a $'s and $ b $'s,
	only the distinguished target remains and the computation goes on in $ \mathsf{path}_{k} $'s
	with probability 1.
	Otherwise, with probability $ 1-\frac{1}{N} $,
	the computation enters $ \mathsf{rejecting\mbox{-}path_{j}} $'s,
	in which the input is rejected certainly at the end.
	The same procedure is repeated for each $ (a^{+}b^{+}) $ block that is followed by an $ a $	
	in $ \mathsf{path}_{k} $'s.
	When reading $ \dollar $, the computation again enters a final QFT in $ \mathsf{path}_{k} $'s such that
	the distinguished target is a configuration with the accepting state $ p_{N} $.
	
	Therefore, the members of $ L_{upal}^{*} $ are accepted exactly and
	the nonmebers are rejected with probability at least $ 1 - \frac{1}{N} $.
	By setting $ N = \left\lceil \frac{1}{\epsilon} \right\rceil $,
	we obtain the desired machine.
\end{proof}

	In \cite{Gr78},
	it was shown that $ L_{upal}^{*} $ cannot not recognized by any 1D$ k $BCAs,
	where $ k \in \mathbb{Z}^{+} $.
	
\begin{corollary}
	For any $ k \in \mathbb{Z}^{+} $ and $ \epsilon \in (0,\frac{1}{2}) $,
	the class of languages recognized by 1D$ k $BCAs is a proper subset of 
	the class of the languages recognized by rtQ$ k $BCAs with error bound $ \epsilon $.
\end{corollary}

\begin{conjecture}
	$ L_{upal}^{*} $ cannot be recognized by any 1P$ k $BCA with bounded error, where $ k>0 $.
\end{conjecture}

% ssssssssssssssssssssssssssssssssssssssssssssssssssssssssssssssssssssssssssssssss %
% ssssssssssssssssssssssssssssssssssssssssssssssssssssssssssssssssssssssssssssssss %
\subsection{Multihead finite automata}
% ssssssssssssssssssssssssssssssssssssssssssssssssssssssssssssssssssssssssssssssss %
% ssssssssssssssssssssssssssssssssssssssssssssssssssssssssssssssssssssssssssssssss %

Let $ L_{upal(t)} $ and $ L'_{upal(t)} $ be the following languages:
\begin{equation*}
	L_{upal(t)} = \{ a^{n_{1}} b  \cdots b a^{n_{t}} b a^{n_{t}} b \cdots b a^{n_{1}} 
	\mid n_{i} \geq 0, 1 \le i \le t \}
\end{equation*}
and
\begin{equation*}
	L'_{upal(t)} = \{ a^{n_{1}} b  \cdots b a^{n_{t}} b a^{n_{t}} b \cdots b a^{n_{1}} 
	\mid n_{i} > 0, 1 \le i \le t \},
\end{equation*}
respectively.
It was shown in \cite{Ku91} that for any $ k $, there exists a $ t>0 $ such that
language $ L'_{upal(t)} $ cannot be recognized by any 1N$ k $FA.
We can argue the same argument also for $ L_{upal(t)} $ since
any 1N$ k $FA recognizing $ L_{upal(t)} $ can be converted to a 1N$ k $FA recognizing $ L'_{upal(t)} $
in a straightforward way.
On the other hand, we show that, for any $ \epsilon \in (0,\frac{1}{2}) $,
$ L_{upal(t)} $ can be recognized by a 1QFA or a 1P3FA with negative one-sided error bound $ \epsilon $.

\begin{lemma}
	\label{lem:upal(1)-epsilon}
	For any $ \epsilon \in (0,\frac{1}{2}) $, language $ L_{upal(1)} $ can be recognized by a 1QFA
	with negative one-sided error bound $ \epsilon $.
\end{lemma}
\begin{proof}
	We use a similar technique described in the proof of Lemma \ref{lem:upal-epsilon}.
	Let $ N = \left\lceil \frac{1}{\epsilon} \right\rceil $ and 
	$ \mathcal{M}_{upal(1),N} = (Q,\Sigma,\Omega,\delta,q_{0},\Omega_{a},\Omega_{r}) $ 
	be a 1QFA, where
	$ Q= \{ \stopstate{q_{k}} \mid 0 \le k \le N\} \cup 
	\{ \rightstate{q_{j,1}} \cup \rightstate{p_{j,1}} \mid 1 \leq j \leq N \}
	\cup
	\{ \stopstate{q_{j,k}} \mid 2 \leq k \leq j+1 , 1 \leq j \leq N \}
	\cup
	\{ \stopstate{p_{j,k}} \mid 2 \leq k \leq N-j+2 , 1 \leq j \leq N \} $,
	$ \Omega = \{\omega_{n},\omega_{a},\omega_{r}\} $,
	$ \Omega_{a}=\{\omega_{a}\} $, and $ \Omega_{r} = \{ \omega_{r} \} $.
	The details of $ \delta $ is given in Figure \ref{fig:M-upal(1)-epsilon}. 
	(The missing part of $ \delta $ can be easily be completed.)
	
	We show that $ \mathcal{M}_{upal(1),N} $ recognizes $ L_{upal(1)} $
	with negative one-sided error bound $ \frac{1}{N} $.
	Therefore, by setting $ N = \left\lceil \frac{1}{\epsilon} \right\rceil $,
	we obtain the desired machine.
	
	\begin{figure}[h!]
	\centering
	\scriptsize
	\fbox{
	\begin{minipage}{0.96\textwidth}		
		\begin{equation*}
			\begin{array}{|c|c|}
				\hline
				\multicolumn{2}{|c|}{\cent}
				\\ \hline
				\multicolumn{2}{|l|}{
					\begin{array}{lcl}
						\delta(\rightstate{q_{0}},\cent) & = & \frac{1}{\sqrt{N}} ( \rightstate{q_{j,1}},\omega_{n})
					\end{array}
					}
				\\ \hline \multicolumn{2}{c}{~} \\
				\multicolumn{2}{c}{\mathsf{path_{j}}~(1 \leq j \leq N)}
				\\ \hline				
					 a~(\mbox{\tiny before reading a $ b $}) & a~(\mbox{\tiny after reading a $ b $})
				\\ \hline
					\begin{array}{lcl}						
						\delta(\rightstate{q_{j,1}},a) & = & (\stopstate{q_{j,2}},\omega_{n}) \\
						\delta(\stopstate{q_{j,k}},a) & = & 
							(\stopstate{q_{j,k+1}},\omega_{n}) ~~ (2 \le k < j+1) \\
						\delta(\stopstate{q_{j,j+1}},a) & = & (\rightstate{q_{j,1}},\omega_{n}) \\						
					\end{array}
				&
					\begin{array}{lcl}
						\delta(\rightstate{p_{j,1}},a) & = & (\stopstate{p_{j,2}},\omega_{n}) \\
						\delta(\stopstate{q_{j,k}},a) & = & 
							(\stopstate{p_{j,k+1},a},\omega_{n}) ~~ (2 \le k < N-j+2) \\
						\delta(\stopstate{p_{j,N-j+2}},a) & = & (\rightstate{p_{j,1}},\omega_{n}) \\
					\end{array}
				\\ \hline 
				b & \dollar
				\\ \hline
					\begin{array}{lcl}
						\delta(\rightstate{q_{j,1}},b) & = & (\rightstate{p_{j,1}},\omega_{n}) \\
						\delta(\rightstate{p_{j,1}},b) & = & (\rightstate{p_{j,1}},\omega_{r})
					\end{array}
				&
					\begin{array}{lcl}
						\delta(\rightstate{q_{j,1}},\dollar) & = & (\stopstate{q_{j,2}},\omega_{r}) \\
						\delta(\rightstate{p_{j,1}},\dollar) & = & \frac{1}{\sqrt{N}} 
							\sum\limits_{l = 1}^{N-1} e^{\frac{2 \pi i }{N}jl} (\stopstate{q_{l}},\omega_{r}) +
								\frac{1}{\sqrt{N}} e^{2 \pi i j} (\stopstate{q_{N}},\omega_{a})
					\end{array}
				\\ \hline
			\end{array}
		\end{equation*}			
	\end{minipage}} 
	\caption{The details of the transition function of $ \mathcal{M}_{upal(1),N} $}
	\label{fig:M-upal(1)-epsilon}
	\end{figure}
	
	On symbol $ \cent $, the computation is split into $ N $ different paths, say $ \mathsf{path}_{i} $
	$ (1 \le i \le N) $, with amplitude $ \frac{1}{\sqrt{N}} $.
	If the input does not contain exactly one $ b $, 
	it is rejected in each path.
	We assume the input of the form $ a^{m}ba^{n} $ ($ m,n \geq 0 $) in the remaining part.
	Before (resp., after) reading symbol $ b $, $ \mathsf{path}_{j} $ waits $ j $ (resp., $ N-j+1 $) step(s) 
	on each $ a $, and so,
	$ \mathsf{path}_{j} $ arrives on $ \dollar $ after making $ m(j)+n(N-j+1) $ stationary movements.
	After reading $ \dollar $, each path makes a QFT: the input is accepted in the distinguished 
	target and it is rejected, otherwise.
	
	It can be easily verified that for any $ j_{1} \neq j_{2} $,
	$ \mathsf{path}_{j_{1}} $ and $ \mathsf{path}_{j_{2}} $ arrive on $ \dollar $ 
	simultaneously if and only if $ m=n $, where $ 1 \le j_{1},j_{2} \le N $.
	In other words, each path makes the $ N $-way QFT
	at the same time if and only if $ m=n $. 
	That is, (i) if $ m=n $ (\textit{the succeed case}), all paths are interfered with each other
	and only configuration $ \ket{(q_{N},|w|+2)} $ remains with probability 1 and so the input is accepted exactly;
	(ii) otherwise (\textit{the failure case}), none of the paths is interfered with the others and so 
	the input is accepted with probability at most $ \frac{1}{N} $.
\end{proof}

\begin{theorem}
	\label{thm:upal(t)-epsilon}
	For any $ \epsilon \in (0,\frac{1}{2}) $, language $ L_{upal(t)} $ can be recognized by a 1QFA
	with negative one-sided error bound $ \epsilon $, where $ t>0 $.
\end{theorem}
\begin{proof}
	(Sketch)
	The proof can be obtained by generalizing the technique presented 
	in the proof of Lemma \ref{lem:upal(1)-epsilon}.
	Suppose that the input is of the form 
	\begin{equation*}
		\label{eq:upal(t)}
		a^{m_{1}} b \cdots b a^{m_{t}} b a^{n_{t}} b \cdots b a^{n_{1}} ~
			(m_{i},n_{i} \geq 0 \mbox{ and } 1 \le i \le t).
	\end{equation*}
	(Otherwise, the input is rejected exactly.)
	The algorithm has $ t $ stages.
	In the first stage, the equality of $ m_{t} $ and $ n_{t} $ are compared. 
	If so, the computation goes to next stage with probability 1.
	Otherwise, the input is rejected with probability $ 1-\frac{1}{N} $ and 
	the computation goes to next stage with probability $ \frac{1}{N} $.
	In the second stage, the equality of $ m_{t-1} $ and $ n_{t-1} $ are compared
	in the same manner. The computation continues in this way and
	in the last stage, the input is accepted instead of going to next stage.
	Therefore, for the members, the input is accepted with probability 1 and for the nonmembers
	the input is accepted with probability at most $ \frac{1}{N} $.	
	Some technical details are given below. Let $ N = \left\lceil \frac{1}{\epsilon} \right\rceil $.
	
	1. On symbol $ \cent $,
			the computation is split into $ N $ paths with equal amplitudes,
			say $ \mathsf{path}_{j_{1}} $ ($ 1 \le j_{1} \le N $).
			After reading the first $ b $, the computation is again split into $ N $ paths with equal amplitudes,
			i.e. $ \mathsf{path}_{j_{1}} $ is split into $ N $ paths 
			$ \mathsf{path}_{j_{1},j_{2}} $  $ (1 \le j_{1},j_{2} \le N) $.
			This process is repeated until reading the $ (t-1)^{th} $ $ b $.
			Thus, after reading the $ (t-1)^{th} $ $ b $, each path has $ t $ indexes,
			i.e. $ \mathsf{path}_{j_{1},\ldots,j_{t}} $ ($ 1 \le j_{k} \le N $ and $ 1 \le k \le t $).
			Note that, any path with index $ (j_{1},j_{2},\ldots,j_{k'}) $ ($ 1 \le k' \le t $)
			is responsible to compare numbers $ m_{k'} $ and $ n_{k'} $.
	
	2.  Before (resp., after) reading the $ t^{th} $ $ b $, if $ j $ is the last one in 
			the index (of the path), then, it waits $ j $ (resp., $ N-j+1 $) steps over each $ a $.

	3. After reading the $ t^{th} $ $ b $, all paths start to make $ N $-way QFT over each $ b $ 
			in order to compare the numbers under their responsibility. After the QFT,
			the computation continues with the paths,
			from which the current paths were created in the previous steps
			(i.e. technically the rightmost one is dropped from the index)
			with probability 1 in the succeed case and with probability $ \frac{1}{N} $
			in the failure case.
			Note that, in the failure case, the computation is terminated and the input is rejected
			with probability $ 1-\frac{1}{N} $.
\end{proof}

\begin{corollary}
	For any $ k \in \mathbb{Z}^{+} $ and $ \epsilon \in (0,\frac{1}{2}) $,
	the class of languages recognized by 1D$ k $FAs is a proper subset of 
	the class of the languages recognized by 1Q$ k $FAs with error bound $ \epsilon $.
\end{corollary}			
			
In \cite{Fr79}, Freivalds showed that, for any $ \epsilon \in (0,\frac{1}{2}) $, $ L_{eq(t)} $
can be recognized by a rtP1BCA with negative one-sided error bound $ \epsilon $, where $ t>0 $ and
\begin{equation*}
	L_{eq(t)}=\{ w \in \{a_{1},\ldots,a_{t},b_{1},\ldots,b_{t}\}^{*} \mid \forall i \in \{1,\ldots,t\} 
	(|w|_{a_{i}} = |w|_{b_{i}}) \}.
\end{equation*}
In fact, it is not hard to modify the Freivalds' algorithm in order to show that,
for any $ \epsilon \in (0,\frac{1}{2}) $,
rtP1BCA can recognize $ L_{upal(t)} $ with negative one-sided error bound $ \epsilon $.
Moreover, since the task of any counter can be implemented by two heads, we can argue the following result.
\begin{lemma}
	\label{lem:rtP1BCA-by-1P3FA}
	Any rtP1BCA can be exactly simulated by a 1P3FA.
\end{lemma}
\begin{proof}
	Let $ \mathcal{M} $ and $ \mathcal{M'} $ be respectively the given rtP1BCA and the 1P3FA 
	simulating $ \mathcal{M} $.
	The heads of $ \mathcal{M'} $ can be named as follows:
	$ H_{i} $ is the head simulating the input head of $ \mathcal{M} $ and 
	$ H_{1} $ and $ H_{2} $ are responsible to implement the blind counter of $ \mathcal{M} $.
	The input is sequentially read by $ H_{i} $ as $ \mathcal{M} $ does and
	for any increment (resp., decrement) operation on the counter, 
	$ H_{1} $ (resp., $ H_{2} $) moves one square to the right. 
	When $ H_{i} $ reads the right end-marker and enters an accepting state,
	both $ H_{1} $ and $ H_{2} $ are tested whether they are on the same square or not
	(they start to travel towards to the right end-marker ($ \dollar $) 
	with the same speed and	the test is passed if they read $ \dollar $ simultaneously).
	If so, the input is accepted, otherwise, it is rejected.
\end{proof}

\begin{theorem}
	\label{thm:upal(t)-1P3FA}
	For any $ \epsilon \in (0,\frac{1}{2}) $, language $ L_{upal(t)} $ can be recognized by a 1P3FA
	with negative one-sided error bound $ \epsilon $, where $ t>0 $.
\end{theorem}

\begin{corollary}
	For any $ k \geq 3 $ and $ \epsilon \in (0,\frac{1}{2}) $,
	the class of languages recognized by 1D$ k $FAs is a proper subset of 
	the class of the languages recognized by 1P$ k $FAs with error bound $ \epsilon $.
\end{corollary}

By using $ t $ heads, it is not hard to show that a 1QFA can recognize language $ L_{neq(t)} $
with any error bound less than $ \frac{1}{3} $, where
\begin{equation*}
	L_{neq(t)}=\{ w \in \{a_{1},\ldots,a_{t},b_{1},\ldots,b_{t}\}^{*} \mid \forall i \in \{1,\ldots,t\} 
	(|w|_{a_{i}} \neq |w|_{b_{i}}) \}.
\end{equation*}

\begin{question}
	What is the minimum number of heads required by a 1PFA in order to recognize $ L_{neq(t)} $
	with an error bound less than $ \frac{1}{3} $?
\end{question}

By using the techniques described in Section \ref{sec:nondeterminism}, we can show that
$ L_{neq(t)} $ is a member of NQAL, where $ t>1 $.

\begin{conjecture}
	For any $ k > 1 $, there exists a $ t > 0 $ such that $ L_{neq(t)} $ cannot be recognized
	by any 1N$ k $FA.
\end{conjecture}

\begin{conjecture}
	$ L_{gt} = \{ w \in \{a,b\}^{*} \mid |w|_{a} > |w|_{b} > 0 \} $ 
	cannot be recognized by a 1QFA with bounded error?
\end{conjecture}

\begin{question}
	What is the minimum number of heads required by a 1QFA (or a 1PFA) in order to recognize $ L_{gt(t)} $
	with an error bound less than $ \frac{1}{3} $, where $ t>1 $ and
	$ L_{gt(t)}=\{ w \in \{a_{1},\ldots,a_{t},b_{1},\ldots,b_{t}\}^{*} \mid \forall i \in \{1,\ldots,t\} 
		(|w|_{a_{i}} > |w|_{b_{i}}) \} $?
\end{question}

% ssssssssssssssssssssssssssssssssssssssssssssssssssssssssssssssssssssssssssssssss %
% ssssssssssssssssssssssssssssssssssssssssssssssssssssssssssssssssssssssssssssssss %
\subsection{Multihead pushdown automata}
% ssssssssssssssssssssssssssssssssssssssssssssssssssssssssssssssssssssssssssssssss %
% ssssssssssssssssssssssssssssssssssssssssssssssssssssssssssssssssssssssssssssssss %

It was shown in \cite{CL88} that $ L_{twin(t)} $, namely \textit{twin languages}, 
cannot be recognized by a 1NPD$ k $A
if and only if $ t > \scriptsize \left( \begin{array}{cc} k \\ 2 \end{array} \right) $, where
$ t>0 $, $ k > 1 $, and
\begin{equation*}
	L_{twin(t)} = \{ w_{1} c \cdots c w_{t} c w_{t} c \cdots c w_{1} \mid w_{i} \in \{a,b\}^{*}, 1 \le i \le t \}.
\end{equation*}
Note that, $ L_{twin(t)} $ can be recognized by a 1D$ k $FA whenever 
$ t \leq \scriptsize \left( \begin{array}{cc} k \\ 2 \end{array} \right) $ \cite{Ro66} and so
for this language neither nondeterminism nor a pushdown storage is helpful.

\begin{theorem}
	\label{thm:L-twin(2t)-1PkFA}
	$ L_{twin(2t)} $ can be recognized by a 1P$ k $FA
	with negative one-sided error bound $ \frac{1}{2} $,
	where $ t = \scriptsize \left( \begin{array}{cc} k \\ 2 \end{array} \right) $,
	$ t>0 $, $ k > 1 $.
\end{theorem}
\begin{proof}
	We assume the input of the form  -- if not, it is rejected exactly --
	\begin{equation*}
		w_{1} c \ldots w_{2t} c u_{2t} c \ldots c u_{1} ~
		(w_{i},u_{i} \in \{a,b\}^{*}, 1 \leq i \leq 2t).
	\end{equation*}	
	At the beginning, the computation is split into two branches, 
	say $ \mathsf{branch_{1}} $ and $ \mathsf{branch_{2}} $, with probability $ \frac{1}{2} $.
	By using $ k $ heads, the pairs $ (w_{1},u_{1}),\ldots,(w_{t},u_{t}) $ 
	and $ (w_{t+1},u_{t+1}),\ldots,(w_{2t},u_{2t}) $
	are compared deterministically in $ \mathsf{branch_{1}} $ and $ \mathsf{branch_{2}} $,
	respectively.
\end{proof}

\begin{corollary}
	The class of the languages recognized by 1D2FAs is a proper subset of 
	the class of the languages recognized by 1P2FAs with error bound $ \frac{1}{3} $.
\end{corollary}

\begin{corollary}
	For any $ k \in \mathbb{Z}^{+} $,
	the class of languages recognized by 1DPD$ k $As is a proper subset of 
	the class of the languages recognized by 1QPD$ k $As (1PPD$ k $As) with error bound $ \frac{1}{3} $.		
\end{corollary}

It is an open problem whether $ L_{twin} (= L_{twin(1)}) $ can be recognized by a 1PPDA with bounded error
\cite{YFSA11A}. Therefore, it is interesting to ask the following question.
\begin{question}
	For a given $ t >0 $,
	let $ k $ be the minimum integer such that $ L_{twin(t)} $ is recognized by a 1P$ k $FA 
	with an error bound at most $ \frac{1}{3} $.
	Is there any $ k' < k $ such that $ L_{twin(t)} $ can be recognized by a 1PPD$ k' $A
	with error bound $ \frac{1}{3} $?
\end{question}

It was shown in \cite{YFSA11A} that $ L_{twin} $ can be recognized by a rtQPDA with
negative one-sided bounded error $ \frac{1}{2} $. 
Therefore, a quantum machine can make one more comparison of a pair by using a pushdown storage 
for any twin language.

\begin{corollary}
	\label{cor:L-twin(t)-1QkPDA}
	For a given $ t >0 $,
	let $ k $ be the minimum number such that $ L_{twin(t)} $ is recognized by a 1P$ k $FA 
	with negative one-sided error bound $ \frac{1}{2} $.
	Then, $ L_{twin(t)} $ can be recognized by a 1QPD$ (k\mbox{-}1) $A  with negative one-sided
	error bound $ \frac{1}{2} $.
\end{corollary}

\begin{question}
	Is the class of languages recognized by 1PPD$ k $As with error bound $ \frac{1}{3} $
	is properly contained in the class of languages recognized by 1QPD$ k $As with error bound $ \frac{1}{3} $?
\end{question}

If we allow the error bound bigger than $ \frac{1}{3} $, we can obtain the following results.

\begin{theorem}
	\label{thm:L-twin(t)-rtQPDA}
	$ L_{twin(t)} $ can be recognized by a rtQPDA
	with negative one-sided error bound $ 1-\frac{1}{2t} $, where $ t>0 $.
\end{theorem}
\begin{proof}
	We assume the input of the form -- if now, it is rejected exactly --
	\begin{equation*}
		\label{eq:twin(t)}
		w_{1} c \cdots c w_{t} c u_{t} c \cdots c u_{1}
			~ ( w_{i},u_{i} \in \{a,b\}^{*}, 1 \le i \le t).	
	\end{equation*}
	An integer, say $ i $, is selected from the set $ \{1,\ldots,t\} $ 
	with probability $ \frac{1}{t} $ at the beginning.
	Then, the substrings $ w_{i} $ and $ u_{i} $ are compared by the rtQPDA algorithm
	for $ L_{twin} $ given in \cite{YFSA11A}.
\end{proof}
\begin{theorem}
	\label{thm:L-twin(t)-1P2FA}
	$ L_{twin(t)} $ can be recognized by a 1P2FA (or 1Q2FA)
	with negative one-sided error bound $ 1-\frac{1}{t} $, where $ t>0 $.
\end{theorem}

\textbf{Acknowledgment.}
We thank Juraj Hromkovi\v{c} for his helpful answers to our questions regarding blind counters.

% BBBBBBBBBBBBBBBBBBBBBBBBBBBBBBBBBBBBBBBBBBBBBBBBBBBBBBBBBBBBBBBBBBBBBBBBBBBBBBBB %
% BBBBBBBBBBBBBBBBBBBBBBBBBBBBBBBBBBBBBBBBBBBBBBBBBBBBBBBBBBBBBBBBBBBBBBBBBBBBBBBB %
\bibliographystyle{plain}
\bibliography{YakaryilmazSay}

\begin{thebibliography}{10}

\bibitem{AF98}
Andris Ambainis and R\={u}si\c{n}\v{s} Freivalds.
\newblock 1-way quantum finite automata: strengths, weaknesses and
  generalizations.
\newblock In {\em FOCS'98: Proceedings of the 39th Annual Symposium on
  Foundations of Computer Science}, pages 332--341, 1998.

\bibitem{ANTV02}
Andris Ambainis, Ashwin Nayak, Amnon Ta-Shma, and Umesh Vazirani.
\newblock Dense quantum coding and quantum finite automata.
\newblock {\em Journal of the ACM}, 49(4):496--511, 2002.

\bibitem{AW02}
Andris Ambainis and John Watrous.
\newblock Two--way finite automata with quantum and classical states.
\newblock {\em Theoretical Computer Science}, 287(1):299--311, 2002.

\bibitem{AY11A}
Andris Ambainis and Abuzer Yakary{\i}lmaz.
\newblock {\em Automata: from Mathematics to Applications}, chapter Automata
  and quantum computing.
\newblock (In preparation).

\bibitem{BC01B}
Alberto Bertoni and Marco Carpentieri.
\newblock Analogies and differences between quantum and stochastic automata.
\newblock {\em Theoretical Computer Science}, 262(1-2):69--81, 2001.

\bibitem{BMP05}
Alberto Bertoni, Carlo Mereghetti, and Beatrice Palano.
\newblock Small size quantum automata recognizing some regular languages.
\newblock {\em Theoretical Computer Science}, 340(2):394--407, 2005.

\bibitem{BFK01}
Richard Bonner, R\={u}si\c{n}\v{s} Freivalds, and Maksim Kravtsev.
\newblock Quantum versus probabilistic one-way finite automata with counter.
\newblock In {\em SOFSEM 2007: Theory and Practice of Computer Science}, volume
  2234 of {\em Lecture Notes in Computer Science}, pages 181--190, 2001.

\bibitem{BP02}
Alex Brodsky and Nicholas Pippenger.
\newblock Characterizations of 1--way quantum finite automata.
\newblock {\em SIAM Journal on Computing}, 31(5):1456--1478, 2002.

\bibitem{CL88}
Marek Chrobak and Ming Li.
\newblock $ k+1 $ heads are better than $ k $ for $ \mbox{PDA} $s.
\newblock {\em Journal of Computer and System Sciences}, 37:144--155, 1988.

\bibitem{FMR68}
Patrick~C. Fischer, Albert~R. Meyer, and Arnold~L. Rosenberg.
\newblock Counter machines and counter languages.
\newblock {\em Mathematical Systems Theory}, 2(3):265--283, 1968.

\bibitem{Fr79}
R\={u}si\c{n}\v{s} Freivalds.
\newblock Fast probabilistic algorithms.
\newblock In {\em Mathematical Foundations of Computer Science 1979}, volume~74
  of {\em LNCS}, pages 57--69, 1979.

\bibitem{FYS10A}
R\={u}si\c{n}\v{s} Freivalds, Abuzer Yakary{\i}lmaz, and A.~C.~Cem Say.
\newblock A new family of nonstochastic languages.
\newblock {\em Information Processing Letters}, 110(10):410--413, 2010.

\bibitem{Go00}
Marats Golovkins.
\newblock Quantum pushdown automata.
\newblock In {\em SOFSEM'00: Proceedings of the 27th Conference on Current
  Trends in Theory and Practice of Informatics}, pages 336--346, 2000.

\bibitem{Gr78}
S.~A. Greibach.
\newblock Remarks on blind and partially blind one-way multicounter machines.
\newblock {\em Theoretical Computer Science}, 7:311--324, 1978.

\bibitem{Hi08}
Mika Hirvensalo.
\newblock Various aspects of finite quantum automata.
\newblock In {\em DLT'08: Proceedings of the 12th international conference on
  Developments in Language Theory}, pages 21--33, 2008.

\bibitem{HS10}
Juraj Hromkovi\v{c} and Georg Schnitger.
\newblock On probabilistic pushdown automata.
\newblock {\em Information and Computation}, 208(8):982--995, 2010.

\bibitem{KW97}
Attila Kondacs and John Watrous.
\newblock On the power of quantum finite state automata.
\newblock In {\em FOCS'97: Proceedings of the 38th Annual Symposium on
  Foundations of Computer Science}, pages 66--75, 1997.

\bibitem{Kr99}
Maksim Kravtsev.
\newblock Quantum finite one-counter automata.
\newblock In {\em SOFSEM'99: Theory and Practice of Computer Science}, volume
  1725 of {\em Lecture Notes in Computer Science}, pages 431--440, 1999.

\bibitem{Ku91}
Miros{\l}aw Kuty{\l}owski.
\newblock Multihead one-way finite automata.
\newblock {\em Theoretical Computer Science}, 85(1):135--153, 1991.

\bibitem{NIH01}
Kiyoharu~Hamaguchi Masaki~Nakanishi, Takao~Indoh.
\newblock On the power of quantum pushdown automata with a classical stack and
  1.5-way quantum finite automata.
\newblock Technical Report NAIST-IS-TR2001005, Nara Institute of Science and
  Technology, 2001.
\newblock http://isw3.naist.jp/IS/TechReport/report/2001005.ps.

\bibitem{MP02}
Carlo Mereghetti and Beatrice Palano.
\newblock On the size of one-way quantum finite automata with periodic
  behaviors.
\newblock {\em Theoretical Informatics and Applications}, 36(3):277--291, 2002.

\bibitem{MC00}
Cristopher Moore and James~P. Crutchfield.
\newblock Quantum automata and quantum grammars.
\newblock {\em Theoretical Computer Science}, 237(1-2):275--306, 2000.

\bibitem{MNYW05}
Yumiko Murakami, Masaki Nakanishi, Shigeru Yamashita, and Katsumasa Watanabe.
\newblock Quantum versus classical pushdown automata in exact computation.
\newblock {\em IPSJ Digital Courier}, 1:426--435, 2005.

\bibitem{NHK06}
Masaki Nakanishi, Kiyoharu Hamaguchi, and Toshinobu Kashiwabara.
\newblock Expressive power of quantum pushdown automata with classical stack
  operations under the perfect-soundness condition.
\newblock {\em IEICE - Transactions on Information and Systems},
  E89-D(3):1120--1127, 2006.

\bibitem{NC00}
Michael~A. Nielsen and Isaac~L. Chuang.
\newblock {\em Quantum Computation and Quantum Information}.
\newblock Cambridge University Press, 2000.

\bibitem{Pa71}
Azaria Paz.
\newblock {\em Introduction to Probabilistic Automata}.
\newblock Academic Press, New York, 1971.

\bibitem{Ro66}
Arnold~L. Rosenberg.
\newblock On multi-head finite automata.
\newblock {\em IBM Journal of Research and Development}, 10(5):388--394, 1966.

\bibitem{SY10A}
A.~C.~Cem Say and Abuzer Yakary{\i}lmaz.
\newblock Quantum function computation using sublogarithmic space, 2010.
\newblock (Poster presentation at QIP2010) (arXiv:1009.3124).

\bibitem{SY11A}
A.~C.~Cem Say and Abuzer Yakary{\i}lmaz.
\newblock Computation with narrow $ \mbox{CTCs} $.
\newblock In {\em Unconventional Computation}, volume 6714 of {\em Lecture
  Notes in Computer Science}, pages 201--211, 2011.

\bibitem{SY11C}
A.~C.~Cem Say and Abuzer Yakary{\i}lmaz.
\newblock Quantum counter automata.
\newblock Technical Report arXiv:1105.0165, 2011.
\newblock (A preliminary version of this paper appeared in the
  \textit{Proceedings of Randomized and Quantum Computation (satellite workshop
  of MFCS and CSL 2010)}, pages 25--34, 2010).

\bibitem{Si06}
Michael Sipser.
\newblock {\em Introduction to the Theory of Computation, 2nd edition}.
\newblock Thomson Course Technology, United States of America, 2006.

\bibitem{Tu69}
Paavo Turakainen.
\newblock Generalized automata and stochastic languages.
\newblock {\em Proceedings of the American Mathematical Society}, 21:303--309,
  1969.

\bibitem{Wa09}
John Watrous.
\newblock Quantum computational complexity.
\newblock In Robert~A. Meyers, editor, {\em Encyclopedia of Complexity and
  Systems Science}, pages 7174--7201. Springer, 2009.

\bibitem{Wa09B}
John Watrous, personal communication, May 2009.

\bibitem{Ya11A}
Abuzer Yakary{\i}lmaz.
\newblock {\em Classical and Quantum Computation with Small Space Bounds}.
\newblock PhD thesis, Bo\u{g}azi\c{c}i University, 2011.
\newblock (arXiv:1102.0378).

\bibitem{YFSA11A}
Abuzer Yakary{\i}lmaz, R\={u}si\c{n}\v{s} Freivalds, A.~C.~Cem Say, and Ruben
  Agadzanyan.
\newblock Quantum computation with write-only memory.
\newblock {\em Natural Computing}.
\newblock (To appear) (arXiv:1011.1201).

\bibitem{YS09B}
Abuzer Yakary{\i}lmaz and A.~C.~Cem Say.
\newblock Efficient probability amplification in two-way quantum finite
  automata.
\newblock {\em Theoretical Computer Science}, 410(20):1932--1941, 2009.

\bibitem{YS09C}
Abuzer Yakary{\i}lmaz and A.~C.~Cem Say.
\newblock Languages recognized with unbounded error by quantum finite automata.
\newblock In {\em CSR'09: Proceedings of the Fourth International Computer
  Science Symposium in Russia}, volume 5675 of {\em Lecture Notes in Computer
  Science}, pages 356--367, 2009.

\bibitem{YS10A}
Abuzer Yakary{\i}lmaz and A.~C.~Cem Say.
\newblock Languages recognized by nondeterministic quantum finite automata.
\newblock {\em Quantum Information and Computation}, 10(9\&10):747--770, 2010.

\bibitem{YS10B}
Abuzer Yakary{\i}lmaz and A.~C.~Cem Say.
\newblock Succinctness of two-way probabilistic and quantum finite automata.
\newblock {\em Discrete Mathematics and Theoretical Computer Science},
  12(2):19--40, 2010.

\bibitem{YS11B}
Abuzer Yakary{\i}lmaz and A.~C.~Cem Say.
\newblock Probabilistic and quantum finite automata with postselection.
\newblock Technical Report arXiv:1102.0666, 2011.
\newblock (A preliminary version of this paper appeared in the
  \textit{Proceedings of Randomized and Quantum Computation (satellite workshop
  of MFCS and CSL 2010)}, pages 14--24, 2010).

\bibitem{YS11A}
Abuzer Yakary{\i}lmaz and A.~C.~Cem Say.
\newblock Unbounded-error quantum computation with small space bounds.
\newblock {\em Information and Computation}, 279(6):873--892, 2011.

\bibitem{YKI05}
Tomohiro Yamasaki, Hirotada Kobayashi, and Hiroshi Imai.
\newblock Quantum versus deterministic counter automata.
\newblock {\em Theoretical Computer Science}, 334(1-3):275--297, 2005.

\bibitem{YKTI02}
Tomohiro Yamasaki, Hirotada Kobayashi, Yuuki Tokunaga, and Hiroshi Imai.
\newblock One-way probabilistic reversible and quantum one-counter automata.
\newblock {\em Theoretical Computer Science}, 289(2):963--976, 2002.

\end{thebibliography}
% BBBBBBBBBBBBBBBBBBBBBBBBBBBBBBBBBBBBBBBBBBBBBBBBBBBBBBBBBBBBBBBBBBBBBBBBBBBBBBBB %
% BBBBBBBBBBBBBBBBBBBBBBBBBBBBBBBBBBBBBBBBBBBBBBBBBBBBBBBBBBBBBBBBBBBBBBBBBBBBBBBB %

\newpage

\appendix

% SSSSSSSSSSSSSSSSSSSSSSSSSSSSSSSSSSSSSSSSSSSSSSSSSSSSSSSSSSSSSSSSSSSSSSSSSSSSSSSS %
% SSSSSSSSSSSSSSSSSSSSSSSSSSSSSSSSSSSSSSSSSSSSSSSSSSSSSSSSSSSSSSSSSSSSSSSSSSSSSSSS %
% SSSSSSSSSSSSSSSSSSSSSSSSSSSSSSSSSSSSSSSSSSSSSSSSSSSSSSSSSSSSSSSSSSSSSSSSSSSSSSSS %
\section{Proof of Theorem \ref{thm:L-ijk}} \label{app:thm:L-ijk}
% SSSSSSSSSSSSSSSSSSSSSSSSSSSSSSSSSSSSSSSSSSSSSSSSSSSSSSSSSSSSSSSSSSSSSSSSSSSSSSSS %
% SSSSSSSSSSSSSSSSSSSSSSSSSSSSSSSSSSSSSSSSSSSSSSSSSSSSSSSSSSSSSSSSSSSSSSSSSSSSSSSS %
% SSSSSSSSSSSSSSSSSSSSSSSSSSSSSSSSSSSSSSSSSSSSSSSSSSSSSSSSSSSSSSSSSSSSSSSSSSSSSSSS %

\textbf{Theorem \ref{thm:L-ijk}.} $ L_{ijk} $ is in NQAL.
\begin{proof}
	By tensoring two GFAs (see Page 147 in \cite{Pa71})
	$ \mathcal{G}_{1} = (Q_{1},\Sigma,\{A_{\sigma \in \Sigma} \},v_{0},f) $
	and 
	$ \mathcal{G}_{2} = (Q_{2},\Sigma, $ $\{B_{\sigma \in \Sigma} \},u_{0},g) $,
	we obtain a new GFA $ \mathcal{G'} $ ($ \mathcal{G}_{1} \otimes \mathcal{G}_{2} $), specified as
	\begin{equation*}
		\mathcal{G'} =  (Q_{1} \times Q_{2},\Sigma, \{ A_{\sigma} \otimes B_{\sigma} \mid \sigma \in \Sigma \},
		v_{0} \otimes u_{0}, f \otimes g) ,
	\end{equation*}
	such that for any $ w \in \Sigma $, 
	\begin{equation*}
		f_{\mathcal{G'}} (w) = f_{\mathcal{G}_{1}} (w) f _{\mathcal{G}_{2}} (w).
	\end{equation*}
	
	Let $ \Sigma=\{a,b,c\} $ be the input alphabet. We design a GFA $ \mathcal{G}_{a-b} $ to calculate the
	value of ($ |w|_{a}-|w|_{b} $) as its accepting value for any $ w \in \Sigma^{*} $, i.e.
	\begin{equation*}
		\mathcal{G}_{a-b} = (Q,\Sigma,\{A_{\sigma \in \Sigma}\},v_{0},f ),
	\end{equation*}
	where $ Q=\{q_{1},q_{2}\} $, $ v_{0} = ( 0~~1 )^{T} $, $ f = (1 ~~ 0 ) $, and
	\begin{equation*}
		A_{a} = \left( \begin{array}{rr}
			1 & 1 \\
			0 & 1 
		\end{array} \right),
		~~
		A_{b} = \left( \begin{array}{rr}
			1 & -1 \\
			0 & 1
		\end{array} \right),
		~~
		A_{c} = \left( \begin{array}{rr}
			1 & 0 \\
			0 & 1
		\end{array} \right).
	\end{equation*}
	At the beginning, the values of $ q_{1} $ and $ q_{2} $ are 0 and 1, respectively. 
	Whenever an $ a $ (resp., a $ b $) is read,
	the value of $ q_{1} $ (resp., $ q_{2} $) is increased (resp., decreased) by 1. 
	At the end, the value of $ q_{1} $ is assigned as the accepting value.
	That is, $ f_{\mathcal{G}_{a-b}} (w) = |w|_{a}-|w|_{b} $.
	
	Similarly, we can design two GFAs $ \mathcal{G}_{a-c} $ and $ \mathcal{G}_{b-c} $ to calculate the
	values of ($ |w|_{a}-|w|_{c} $) and ($ |w|_{b}-|w|_{c} $) as their accepting values,
	respectively, for any $ w \in \Sigma^{*} $.
	
	Moreover, we can design a GFA $ \mathcal{G}_{a^{+}b^{+}c^{+}} $ to assign 1 as the accepting value
	for the strings of the form $ a^{+}b^{+}c^{+} $ and 0, otherwise:
	\begin{equation*}
		\mathcal{G}_{a^{+}b^{+}c^{+}} = (Q,\Sigma,\{A_{\sigma \in \Sigma}\},v_{0},f )
	\end{equation*}
	where $ Q=\{q_{1},q_{2},q_{3},q_{4}\} $, $ v_{0} = ( 1~~0~~0~~0 )^{T} $, 
	$ f = (0 ~~ 0 ~~ 0 ~~ 1 ) $, and
	\begin{equation*}
		A_{a} = \left( \begin{array}{rrrr}
			0 & 0 & 0 & 0 \\
			1 & 1 & 0 & 0 \\
			0 & 0 & 0 & 0 \\
			0 & 0 & 0 & 0
		\end{array} \right),
		~~
		A_{b} = \left( \begin{array}{rrrrr}
			0 & 0 & 0 & 0 \\
			0 & 0 & 0 & 0 \\
			0 & 1 & 1 & 0 \\
			0 & 0 & 0 & 0
		\end{array} \right),
		~~
		A_{c} = \left( \begin{array}{rrrrr}
			0 & 0 & 0 & 0 \\
			0 & 0 & 0 & 0 \\
			0 & 0 & 0 & 0 \\
			0 & 0 & 1 & 1
		\end{array} \right).
	\end{equation*}
	
	Now, we can obtain a GFA $ \mathcal{G}_{ijk} $ for $ L_{ijk} $ as
	\begin{equation*}
		\mathcal{G}_{ijk} = \mathcal{G}_{a^{+}b^{+}c^{+}} 
		\otimes 
		\left( \mathcal{G}_{a-b} \otimes \mathcal{G}_{a-b} \right)
		\otimes
		\left( \mathcal{G}_{a-c} \otimes \mathcal{G}_{a-c} \right)
		\otimes
		\left( \mathcal{G}_{b-c} \otimes \mathcal{G}_{b-c} \right),
	\end{equation*}
	which calculates the value of 
	\begin{equation*}
		(|w|_{a}-|w|_{b})^{2} (|w|_{a}-|w|_{c})^{2} (|w|_{b}-|w|_{c})^{2}
	\end{equation*}
	for the strings of the form $ a^{+}b^{+}c^{+} $ and
	returns 0, otherwise.
	In other words, 
	$ f_{\mathcal{G}_{ijk}} (w) $ is a positive integer if $ w $ is a member of $ L_{ijk} $ and it is zero if
	$ w $ is not a member of $ L_{ijk} $.
\end{proof}

% SSSSSSSSSSSSSSSSSSSSSSSSSSSSSSSSSSSSSSSSSSSSSSSSSSSSSSSSSSSSSSSSSSSSSSSSSSSSSSSS %
% SSSSSSSSSSSSSSSSSSSSSSSSSSSSSSSSSSSSSSSSSSSSSSSSSSSSSSSSSSSSSSSSSSSSSSSSSSSSSSSS %
% SSSSSSSSSSSSSSSSSSSSSSSSSSSSSSSSSSSSSSSSSSSSSSSSSSSSSSSSSSSSSSSSSSSSSSSSSSSSSSSS %
\section{Proof of Theorem \ref{thm:rtDkBCA}} \label{app:thm:rtDkBCA}
% SSSSSSSSSSSSSSSSSSSSSSSSSSSSSSSSSSSSSSSSSSSSSSSSSSSSSSSSSSSSSSSSSSSSSSSSSSSSSSSS %
% SSSSSSSSSSSSSSSSSSSSSSSSSSSSSSSSSSSSSSSSSSSSSSSSSSSSSSSSSSSSSSSSSSSSSSSSSSSSSSSS %
% SSSSSSSSSSSSSSSSSSSSSSSSSSSSSSSSSSSSSSSSSSSSSSSSSSSSSSSSSSSSSSSSSSSSSSSSSSSSSSSS %

\textbf{Theorem \ref{thm:rtDkBCA}.} 
If  $ L $ is recognized by a rtD$ k $BCA, then $ \overline{L} \in $NQAL, where $ k>0 $.

\begin{proof}
	Without the loose of generality, we can assume that the counter operation(s) of a rtD$ k $BCA, 
	say $ \mathcal{D} $,
	can be determined by the internal state to be entered after each transition.
	Thus, the transitions of $ \mathcal{D} $ can be defined from $ Q \times \tilde{\Sigma} $ to $ Q $,
	where $ Q = \{q_{1},\ldots,q_{n}\} $ is the set of the internal state and $ n>0 $,
	where $ q_{1} $ is the initial state.
	Equivalently, for each $ \sigma \in \tilde{\Sigma} $, we can define a matrix (transition matrix),
	say $ T_{\sigma} $,
	whose columns and rows are indexed by the internal states such that
	the $ (j,i)^{th} $ entry of $ T_{\sigma} $ represents the transition value from state $ q_{i} $ to $ q_{j} $.
	Due to its deterministic nature, these transition matrices are (left) stochastic having zero-one 
	stochastic columns.
	
	The state-transition of $ \mathcal{D} $ can be linearized. For this purpose, we define the following 
	components:
	\begin{itemize}
		\item $ Q_{a} $ is the set of accepting states,
		\item $ v_{0} = (1~~0~~\cdots~~0)^{T} $ is an $ n $-dimensional column vector, and
		\item $ f $ is an $ n $-dimensional row vector such that $ f[i] = 1 $ if $ q_{i} \in Q_{a} $
		and $ f[i] = 0 $ if $ q_{i} \notin Q_{a} $, where $ 1 \le i \le n $.
	\end{itemize}
	That is, for a given input string $ w \in \Sigma^{*} $, the state-transition of $ \mathcal{D} $ is traced
	by a (stochastic) column vector, i.e.
	\begin{equation*}
		v_{i} = T_{\tilde{w}_{i}} v_{i-1}
	\end{equation*}
	and 
	\begin{equation*}
		v_{| \tilde{w} |} = T_{\dollar} T_{w_{|w|}} \cdots T_{w_{1}} T_{\cent} v_{0},
	\end{equation*}
	where $ 1 \le i \le |\tilde{w}| $. 
	It can be easily verified that if $ \mathcal{D} $ enters to $ q_{j} $ at the end of the computation
	if and only if $ v_{\tilde{w}}[j]=1 $, where $ 1 \le j \le n $. 
	(Note that, each intermediate $ v_{i} $ is also a 
	stochastic zero-one vector, where $ 1 \le i \le |\tilde{w}| $.)
	
	Let $ p_{l} $ be the $ l^{th} $ prime ($ 1 \le l \le k $). 
	In the above schema, the counter operations of $ \mathcal{D} $ can be simulated by using a simple
	number-theoretic method:
	when the $ l^{th} $ counter of $ \mathcal{D} $ is updated by $ 1 $ 
	(resp., $ 0 $ or $ -1 $), the nonzero entry of $ v_{i} $ is updated by multiplying with 
	$ p_{l} $ (resp., 1 or $ \frac{1}{p_{l}} $),
	where $ 1 \le l \le k $ and $ 1 \le i \le |\tilde{w}| $.
	This method can be embedded into the transition matrices.
	That is, if the value(s) of counter(s) is (are) updated with respect to 
	$ c \in \{-1,0,1\}^{k} $ when entering state $ q_{j} \in Q $,
	in each $ T_{\sigma \in \tilde{\Sigma}} $,
	the nonzero entries on the $ j^{th} $ row is replaced with
	\begin{equation*}
		\prod_{l=1}^{k} (p_{l})^{c[l]}.
	\end{equation*}
	We denote updated matrices as $ T'_{\sigma \in \tilde{\Sigma}} $.
	\begin{equation}
		\label{eq:case}
		\mbox{ 
		\small
		\begin{minipage}{0.85\textwidth}
			Suppose that the value(s) of the counter(s) is (are) $ C \in \mathbb{Z}^{k} $
			at the end of the computation and the computation ends in $ q_{j} \in Q $ on input $ w \in \Sigma^{*} $.
			Then, it can be verified in a straightforward way that  
			\begin{equation*}
				v_{|\tilde{w}|}[j] = \prod_{l=1}^{k} (p_{l})^{C[l]},
			\end{equation*}
			which is 1 if and only if each counter value is zero.
		\end{minipage} 
		}
	\end{equation}	
	Let $ T''_{\sigma \in \tilde{\Sigma}} $ be $ (n+1) \times (n+1) $-dimensional matrices obtained
	form $ T'_{\sigma} $ as
	\begin{equation*}
		T''_{\sigma} = \left( \begin{array}{c|c}
			& 0 \\
			T'_{\sigma} & \vdots \\
			& 0 \\
			\hline
			0 \cdots 0 & 1
		\end{array} \right).
	\end{equation*}
	We design a GFA $ \mathcal{G} $ based on $ \mathcal{D} $ as follows:
	\begin{equation*}
		\mathcal{G} = (Q',\Sigma,\{A_{\sigma \in \Sigma}\}, v'_{0},f'),
	\end{equation*}
	where $ Q'= Q \cup \{q_{n+1}\} $, 
	\begin{equation*}
		v'_{0} = T''_{\cent} \left( \begin{array}{r}
			1 \\ 0 \\ \vdots \\ 0 \\ -1
		\end{array} \right),
		~~
		f' = ( f \mid 1 ) T''_{\dollar},
		\mbox{ and }
		A_{\sigma} = T''_{\sigma}.
	\end{equation*}
	Hence, by using the scenario given in (\ref{eq:case}), we can verify that
	\begin{itemize}
		\item if $ q_{j} \in Q_{a} $, then $ f_{\mathcal{G}}(w) = \left( \prod_{l=1}^{k} (p_{l})^{C[l]} \right) -1 $,
			i.e.
			\begin{itemize}
				\item $ f_{\mathcal{G}}(w) = 0  $ if each counter value is zero and
				\item $ f_{\mathcal{G}}(w) \neq 0  $ if at least one counter value is not zero;
			\end{itemize}
		\item if $ q_{j} \notin Q_{a} $, then $ f_{\mathcal{G}}(w) = -1 $.
	\end{itemize}
	
	Let $ L $ be the language recognized by $ \mathcal{D} $ and 
	$ \mathcal{G}^{2} = \mathcal{G} \otimes \mathcal{G} $.
	Then, for $ w \in L $ (resp., $ w \notin L $), $ f_{\mathcal{G}^{2}}(w) = 0 $
	(resp., $ f_{\mathcal{G}^{2}}(w) > 0 $).
	Thus, $ \overline{L} \in $ NQAL.
\end{proof}

\end{document}